%% file: main.tex
\documentclass[journal,10pt,letterpaper]{IEEEtran}
\usepackage[letterpaper,margin=1in,columnsep=0.24in]{geometry}
\usepackage{amsmath,amssymb,amsthm,mathtools,bm}
\usepackage{graphicx,booktabs,array,tabularx,multirow}
\usepackage{cite,microtype,url}
\usepackage[hidelinks]{hyperref}
\hypersetup{pdftitle={Order-Optimal Coded Caching With File and Demand Privacy},
 pdfauthor={Han Fang, Nan Liu, Wei Kang},
 pdflang={en-US},
 pdfsubject={Joint privacy, difference secret sharing, and converse bounds},
 pdfkeywords={coded caching, file privacy, demand privacy, information-theoretic security}}
\allowdisplaybreaks[1]
\newtheorem{theorem}{Theorem}
\newtheorem{lemma}{Lemma}
\newtheorem{corollary}{Corollary}

\theoremstyle{remark}
\newtheorem{remark}{Remark}
\theoremstyle{plain}
\newcommand{\F}{\mathbb F}

\newcommand{\cG}{\mathcal G}
\newcommand{\cI}{\mathcal I}
\newcommand{\cJ}{\mathcal J}

\newcommand{\cU}{\mathcal U}
\newcommand{\Rfd}{R_{\mathrm{fd}}^{\star}}
\newcommand{\Rfe}{R_{\mathrm{fe}}^{\star}}
\newcommand{\Mone}{M_{\mathrm{fd}}^{\star}(1)}
\newcommand{\LB}{L_{\mathrm{g}}}
\newcommand{\LA}{L_{\mathrm{m}}}
\newcommand{\LD}{L_{\mathrm{d}}}
\newcommand{\Lall}{L_{\mathrm{all}}}

\title{Order-Optimal Coded Caching With File and Demand Privacy}
\author{Han~Fang, Nan~Liu, and Wei~Kang%
\thanks{Han Fang is with the School of Cyber Science and Engineering, Southeast University, Nanjing, China (e-mail: \href{mailto:h_fang@seu.edu.cn}{h\_fang@seu.edu.cn}).}%
\thanks{Nan Liu is with the National Mobile Communications Research Laboratory and the School of Cyber Science and Engineering, Southeast University, Nanjing, China (e-mail: \href{mailto:nanliu@seu.edu.cn}{nanliu@seu.edu.cn}).}%
\thanks{Wei Kang is with the School of Information Science and Engineering and the School of Cyber Science and Engineering, Southeast University, Nanjing, China (e-mail: \href{mailto:wkang@seu.edu.cn}{wkang@seu.edu.cn}).}}
\begin{document}
\maketitle
\begin{abstract}
We study coded caching with joint file and demand privacy: each user recovers its requested file while learning nothing about the remaining files and the other users' requests jointly. Let $N$ and $K$ be the numbers of files and users, and let $M$ and $R$ denote the cache memory and delivery rate, both normalized by the file size. For every $N,K\ge2$ and every feasible cache size, we give a scheme whose worst-case delivery rate is at most $9/2$ times the optimum. To reduce the memory used for file shares, the scheme secret-shares $N-1$ differences relative to a reference file. Cached masks supply the correction needed to recover the requested file from its difference. For each integer $t\in\{0,\ldots,K-1\}$, the scheme achieves $M=1+(N-1)t/(K-t)$ and $R=K/(t+1)$. To lower-bound the delivery rate, we compare the joint cache entropy of user groups under alternative demand vectors, using one fixed user's privacy constraint. Two choices of groups determine the optimal delivery rate up to constants at the scale $\min\{K,1+(N-1)/(M-1)\}$ for $M>1$. At $M=1$, the exact optimum is $K$. A complementary bound on joint broadcast entropy gives the minimum memory for unit delivery rate, $1+(N-1)(K-1)$, and the exact tradeoff on an interval ending at that memory. All schemes and bounds apply to repeated as well as distinct requests.
\end{abstract}
\input{sections/01_introduction}
\input{sections/02_model}
\input{sections/03_achievability}
\input{sections/04_converse}
\input{sections/05_gap}
\input{sections/06_boundaries}
\input{sections/07_comparisons}
\section{Conclusion}\label{sec:conclusion}
We have given a jointly file- and demand-private caching scheme and lower bounds that determine the optimal delivery rate within a factor $9/2$ for every $N,K\ge2$ and $M\ge1$. Sharing $N-1$ file differences reduces the memory used for file shares, while cached correction terms let each user recover its requested file. The converse compares joint cache entropies across user groups and alternative demands while keeping one user's request fixed. Together, the construction and converse identify the delivery rate scale $\min\{K,1+(N-1)/(M-1)\}$ for $M>1$. At $M=1$, the optimal delivery rate is exactly $K$. A complementary converse determines the minimum memory for unit delivery rate and an exact tradeoff interval ending at that memory.
\appendices
\input{sections/08_appendices}
\input{sections/09_additional_appendices}
\bibliographystyle{IEEEtran}
\bibliography{references}
\end{document}

%% file: sections/01_introduction.tex
\section{Introduction}\label{sec:introduction}
\IEEEPARstart{C}{oded} caching lets a server satisfy several users with one transmission by coordinating what they store before their requests are known. In the shared-link model of Maddah-Ali and Niesen~\cite{man2014}, each user combines the common transmission with its cache to recover a requested file. The same cache and transmission may also reveal files that the user did not request or the requests of other users.

We require both types of information to remain private. \emph{File privacy} protects the contents of all nonrequested files, and \emph{demand privacy} protects the other users' requests. The joint condition allows each user to learn its requested file while keeping the remaining files and requests jointly private. Let $N$ and $K$ denote the numbers of files and users, respectively. We study the minimum worst-case delivery rate under joint privacy for all $N,K\ge2$. Throughout, cache memory $M$ and delivery rate $R$ are normalized by the file size.

\subsection{Related Work}
Secret sharing and delivery keys protect nonrequested files from users in file-private coded caching~\cite{private2018}. The resulting schemes have a factor-$16$ guarantee over a cache range that includes every feasible cache size when $N\ge K$. Secure delivery instead protects the library from an external observer of the broadcast~\cite{secure2015}. To hide user requests, demand-private caching uses virtual-user constructions~\cite{wan2021,gurjar2022} or privacy keys~\cite{yanprivacy2021}. Key superposition combines external content security with demand privacy~\cite{yansuper2021}.

Qi and Ravi~\cite{qi2022} introduced the joint file-and-demand-private problem. They proved that joint privacy implies secrecy of both the library and the demands from an external observer, characterized the exact two-file, two-user tradeoff, and established the achievable points $(1,K)$ and $(N(K-1)+1,1)$. Their optimality result at memory $M=1$ assumes $N\ge2K$. Qi \emph{et al.}~\cite{qi2024} then proposed secret-sharing/private-key and virtual-user baselines, followed by a multilayer construction. Section~II of~\cite{qi2024} restricts attention to mutually distinct requests. In that setting, the paper reports improved low-memory points and the point $(N(K-1),1)$ at unit delivery rate. Table~\ref{tab:positioning} summarizes these results and their demand assumptions.

For file privacy alone, Fang \emph{et al.}~\cite{fang2025} determined the optimum at $M=1$ and the minimum memory at $R=1$ for all $N,K\ge2$, and an optimal small-cache segment for $N=2$ and $K\ge3$. Their minimum memory at unit delivery rate is $(N-1)(K-1)$. Our construction also uses $N-1$ file differences, together with cache corrections that preserve joint file and demand privacy.

\begin{table*}[t]
\caption{Comparison with the two joint-privacy papers. Memory $M$ and delivery rate $R$ are normalized by the file size. Pairs are written as $(M,R)$. In the proposed family, $t\in\{0,\ldots,K-1\}$ is an integer construction parameter.}
\label{tab:positioning}
\centering\small
\renewcommand{\arraystretch}{1.15}
\begin{tabularx}{\textwidth}{@{}>{\raggedright\arraybackslash}p{0.17\textwidth}>{\raggedright\arraybackslash}p{0.18\textwidth}X@{}}
\toprule
Work & Demand vectors & Results\\
\midrule
Qi--Ravi~\cite{qi2022} & All demands & Joint privacy implies external secrecy. Exact two-file, two-user region. Achievable $(1,K)$ and $(N(K-1)+1,1)$. Optimality of $(1,K)$ for $N\ge2K$.\\
Qi \emph{et al.}~\cite{qi2024} & Mutually distinct demands in Sec.~II & Secret-sharing/private-key and virtual-user baselines. Multilayer memory--rate families. Point $(N(K-1),1)$ at unit delivery rate.\\
This work & All demands & Achievable $\bigl(1+(N-1)t/(K-t),K/(t+1)\bigr)$. Converse based on user groups and alternative demands. Factor $9/2$ for all parameters. Exact endpoints and adjacent optimal segments.\\
\bottomrule
\end{tabularx}
\end{table*}

\subsection{Results and Proof Ideas}
The joint-privacy endpoints and achievable families above do not give a uniform approximation to the optimal delivery rate throughout the feasible memory range. We address this gap, including small caches when $N<K$, with a scheme and a converse that determine the optimum within a factor $9/2$ for every $N,K\ge2$ and $M\ge1$.

To reduce the memory used for file shares, we share the $N-1$ differences between the nonreference files and a reference file. Recovering a requested file also requires its reference-file component. Each cache therefore stores masked correction terms that restore this component during decoding. For every integer construction parameter $t\in\{0,\ldots,K-1\}$, let $(M_t,R_t)$ denote the achievable memory--rate pair:
\begin{equation}\label{eq:intro-points}
 M_t=1+\frac{(N-1)t}{K-t},\qquad R_t=\frac{K}{t+1}.
\end{equation}
At the same value of $t$, the memory is $t/(K-t)$ file units smaller than the secret-sharing/private-key point in~\cite[Th.~1]{qi2024}. To prove joint privacy, we show that the distribution of a user's complete observation can be reproduced using only its requested file and index.

To obtain a lower bound that accounts for both the file count and the user count, we compare several demand vectors under the same cache placement. One fixed user, called the \emph{anchor}, requests a reference file in every demand vector. We arrange other users in groups and assign alternative requests for nonreference files to each group. The anchor's privacy keeps the files requested by a group unknown to the anchor, while decoding requires the group's caches and broadcasts to determine those files. Two conditional-entropy identities bound the available information through $M-1$ and $R-1$. Allowing up to $N-1$ users in a group lets several independent files contribute to the converse at each step.

Two choices of groups and demands suffice to prove the factor-$9/2$ guarantee. To determine the memory needed for unit delivery rate, we instead compare two bounds on joint broadcast entropy. Let $\Mone$ denote the minimum memory at unit delivery rate. The complementary converse gives
\begin{equation}\label{eq:intro-endpoint}
 \Mone=1+(N-1)(K-1),
\end{equation}
and determines the optimal tradeoff on an interval ending at this memory. At memory $M=1$, the optimal delivery rate is exactly $K$ for every $N,K\ge2$.

Section~\ref{sec:model} defines the model and the entropy identities. Sections~\ref{sec:achievability}--\ref{sec:gap} give the scheme, the converse, and the constant-factor proof. Sections~\ref{sec:boundaries} and~\ref{sec:numerics} give exact boundary results and numerical comparisons. The appendices provide the remaining proofs and the full-library SPK implementation used for comparison.

%% file: sections/02_model.tex
\section{System Model and Basic Identities}\label{sec:model}
\subsection{Placement, Delivery, and Joint Privacy}
For a positive integer $m$, let $[m]=\{1,\ldots,m\}$, and set $[0]=\varnothing$. Fix integers $N,K\ge2$. A server holds mutually independent, uniformly distributed files $W_1,\ldots,W_N$, each with entropy $B$ bits. Write $W=(W_1,\ldots,W_N)$. For a file-index set $\mathcal S\subseteq[N]$, let $W_{\mathcal S}$ denote the tuple of files in increasing index order. In particular, $W_{2:N}=(W_2,\ldots,W_N)$.

Let $M$ and $R$ be nonnegative real numbers specifying the cache memory and delivery rate, normalized by $B$. User $k\in[K]$ has a cache $Z_k$ with alphabet $\mathcal Z_k$. For a user set $\mathcal S\subseteq[K]$, write $Z_{\mathcal S}=(Z_k:k\in\mathcal S)$ in increasing user-index order. To form the caches before the demands are known, the server generates a random seed $S$ independently of the library and applies placement functions $\phi_k$:
\begin{equation}
 Z_k=\phi_k(W_{[N]},S).
\end{equation}
The cache alphabet satisfies $\log_2|\mathcal Z_k|\le MB$. Private keys, request-masking vectors, and other locally stored control information all count toward this limit.

The random demands $D_1,\ldots,D_K$ are independent and uniform on $[N]$, and independent of the library and the seed $S$. Let $D=(D_1,\ldots,D_K)$ and $D_{-k}=(D_j:j\ne k)$, with users ordered by index. Lowercase $d=(d_1,\ldots,d_K)\in[N]^K$ denotes a fixed demand vector. The random tuple $W_{[N]\setminus\{D_k\}}$ contains the nonrequested files in increasing index order.

For a realized demand vector $d$, let $X_d$ denote the broadcast and $\mathcal X_d$ its alphabet. The server uses a delivery function $\psi_d$ to form
\begin{equation}
 X_d=\psi_d(W_{[N]},S),\qquad \log_2|\mathcal X_d|\le RB.
\end{equation}
The seed $S$ also includes any randomness used for delivery. All public control information is part of $X_d$. User $k$ decodes using $(Z_k,X_d,d_k)$. The demand alphabet is $[N]^K$, so correctness and the delivery rate constraint apply to repeated as well as distinct requests. Privacy is required against each user observing its own cache and the common broadcast. Users do not pool their caches.

Let $X=X_D$ be the broadcast under the random demand vector. For every user $k\in[K]$, correctness and joint privacy require
\begin{align}
 H(W_{D_k}\mid Z_k,X,D_k)&=0,\label{eq:correct}\\
 I(W_{[N]\setminus\{D_k\}},D_{-k};Z_k,X\mid D_k)&=0.
 \label{eq:jointprivacy}
\end{align}
Equation~\eqref{eq:jointprivacy} protects the other demands and the whole tuple of nonrequested files jointly. This is the joint-privacy formulation of~\cite{qi2022}, written conditionally on the observing user's demand. The chain rule applied first to $D_{-k}$ also gives
\[
 I(W_{[N]\setminus\{d_k\}};Z_k,X_d)=0,\qquad \forall k,d.
\]
Every demand vector has positive probability, so each nonnegative mutual-information term conditioned on a fixed demand must vanish. The resulting fixed-demand file-privacy condition is used in the converse.

A pair $(M,R)$ is achievable if, for every real $\eta>0$, sufficiently large files admit a zero-error, perfectly private scheme with cache size at most $(M+\eta)B$ bits and transmission length at most $(R+\eta)B$ bits. For fixed memory $M$, let $\Rfd(M)$ be the infimum of achievable delivery rates. For a nonnegative target delivery rate $r$, define the minimum memory as $M_{\mathrm{fd}}^{\star}(r)=\inf\{M:\Rfd(M)\le r\}$. The arbitrary slack $\eta$ allows control strings of fixed length to have vanishing normalized cost as the files grow. All entropy inequalities below first apply to finite files and then pass to this achievable closure.

\subsection{What a Private User View Can Depend On}
To prove privacy and derive lower bounds, we characterize the conditional distribution of a user's cache and broadcast given the library and demands. The following representation expresses this distribution as a \emph{probability kernel} indexed only by the requested file and index.

\begin{lemma}[Private-view representation]\label{lem:views}
For every user $k\in[K]$ and requested index $i\in[N]$, there is a probability kernel $\mathsf K_{k,i}$ on user observations, indexed by the value of $W_i$. For every library realization $w=(w_1,\ldots,w_N)$ and demand vector $d\in[N]^K$,
\begin{equation}\label{eq:kernel}
 P_{Z_k,X\mid W=w,D=d}
   =\mathsf K_{k,d_k}(\,\cdot\mid w_{d_k}).
\end{equation}
Moreover,
\begin{equation}\label{eq:independences}
 I(W_{[N]};Z_k)=0,\qquad I(W_{[N]},D;X)=0.
\end{equation}
In particular, $I(W_{[N]};X_d)=0$ for every $d$.
\end{lemma}

Equation~\eqref{eq:kernel} says that a user's observation can be simulated from the user's own requested index and file. A \emph{simulator} is a randomized procedure that reproduces the observation distribution using only those inputs. External secrecy of the broadcast in~\eqref{eq:independences} was established in~\cite[Th.~1]{qi2022}. Appendix~\ref{app:views} proves the representation and both independence identities.

To measure the uncertainty remaining after decoding, we subtract the requested file's $B$ bits from the conditional cache and broadcast entropies. The independence identities above make these deductions exact.

\begin{lemma}[Residual entropy bounds]\label{lem:budgets}
Fix a demand vector $d\in[N]^K$ and a user $k\in[K]$, and let $i=d_k$. Define $\delta=M-1$ and $\varepsilon=R-1$. Every feasible code satisfies $M\ge1$, $R\ge1$, and
\begin{align}
 H(Z_k\mid W_i,X_d)&\le\delta B,\label{eq:cache-budget}\\
 H(X_d\mid Z_k,W_i)&\le\varepsilon B.\label{eq:delivery-budget}
\end{align}
More precisely,
\begin{align}
 H(Z_k\mid W_i,X_d)&=H(Z_k\mid X_d)-B,\label{eq:cache-exact}\\
 H(X_d\mid Z_k,W_i)&=H(X_d\mid Z_k)-B.\label{eq:delivery-exact}
\end{align}
\end{lemma}
\begin{proof}
Correctness gives $H(W_i\mid Z_k,X_d)=0$. By Lemma~\ref{lem:views}, $H(W_i\mid X_d)=H(W_i\mid Z_k)=B$. Expanding the relevant joint entropies yields~\eqref{eq:cache-exact} and~\eqref{eq:delivery-exact}. Bounding $H(Z_k)\le MB$ and $H(X_d)\le RB$ gives~\eqref{eq:cache-budget}--\eqref{eq:delivery-budget}. Nonnegativity gives $M,R\ge1$.
\end{proof}

%% file: sections/03_achievability.tex
\section{A Difference-Sharing Scheme}\label{sec:achievability}
To reduce the cache space occupied by file shares, we share only the $N-1$ differences $W_1-W_N,\ldots,W_{N-1}-W_N$. Secret sharing hides the file contents in each user's cached shares. Independent random vectors hide the requests, and cached correction terms turn each decoded difference into the requested file.

\begin{theorem}[Achievable tradeoff]\label{thm:achievable}
For every $N,K\ge2$ and $t\in\{0,\ldots,K-1\}$, the pair
\begin{equation}\label{eq:points}
 M_t=1+\frac{(N-1)t}{K-t},\qquad R_t=\frac K{t+1}
\end{equation}
is achievable with joint file and demand privacy. Let $U_{N,K}(M)$ be the lower convex envelope of these points, extended by $1$ for $M\ge1+(N-1)(K-1)$. Then $\Rfd(M)\le U_{N,K}(M)$.
\end{theorem}

\subsection{Difference Shares and Local Elimination}
Fix $t\in\{0,\ldots,K-1\}$. Let $F$ be the number of shares for each file difference, $z$ the number of those shares stored by each user, and $h$ the number of data packets in each file. With the convention $z=0$ when $t=0$, set
\begin{equation}\label{eq:dimensions}
 \begin{aligned}
 F&=\binom Kt,\qquad z=\binom{K-1}{t-1},\\
 h&=F-z=\binom{K-1}{t}.
 \end{aligned}
\end{equation}
Let $q$ be a prime power, chosen sufficiently large for the matrices below, and let $\F_q$ be the field with $q$ elements. Fix a positive integer $L$. Represent each file as a column of $h$ independent packets in $\F_q^L$. Matrix multiplication is applied separately to each of the $L$ field coordinates. A packet has $L\log_2 q$ bits, and a file has $B=hL\log_2 q$ bits. For binary files, choose $q$ to be a power of two. If the file length is not divisible by $h\log_2q$, padding requires fewer than $h\log_2q$ additional bits. The normalized padding overhead vanishes as the file length grows.

Choose fixed, public matrices $P\in\F_q^{F\times h}$ and $Q\in\F_q^{F\times z}$ such that $[P\ Q]$ is invertible and every $z$ rows of $Q$ form an invertible matrix. A Vandermonde matrix gives this choice over a field with at least $F$ elements: assign its first $z$ polynomial degrees to $Q$ and its remaining degrees to $P$. Polynomial evaluation is the basis of Shamir's secret-sharing scheme~\cite{shamir1979}. Here the first $z$ coefficients carry randomness and the remaining $h$ coefficients carry data. The resulting ramp-sharing scheme hides the data in any $z$ shares and recovers the data from all $F$ shares, as in private caching~\cite{private2018}.

Index the $F$ rows of $P$ and $Q$ by the $t$-subsets $T$ of $[K]$, in a fixed public order. The subscript $\mathcal H$ on a matrix or a column of packets selects the rows with labels in $\mathcal H$, keeping that order. Use $W_N$ as the reference file. For each $n\in[N-1]$, let $U_n$ be a column of $z$ independent uniform packets in $\F_q^L$. Choose all these packets independently across $n$ and independently of the library. Define the column of difference shares by
\begin{equation}\label{eq:shares}
 S_n=P(W_n-W_N)+QU_n.
\end{equation}
Write $S_{n,T}$ for the packet in row $T$ of $S_n$. User $k$ stores the shares whose row label contains $k$. Define its cached and missing row sets by
\begin{equation}
 \begin{aligned}
 \cI_k&=\{T\subseteq[K]:|T|=t,\ k\in T\},\\
 \cJ_k&=\{T\subseteq[K]:|T|=t,\ k\notin T\}.
 \end{aligned}
\end{equation}
Thus $|\cI_k|=z$ and $|\cJ_k|=h$. Conditional on any fixed library, the cached shares are a fixed shift of $Q_{\cI_k}U_n$. Invertibility of $Q_{\cI_k}$ makes this column uniform. Independence of the columns $U_n$ makes all cached share packets jointly uniform across $n<N$. When $t=0$, $\cI_k$ is empty and there are no cached share packets.

To recover a file difference from the cached and delivered shares, we first eliminate the sharing randomness. When $z=0$, take the inverse of the empty $0\times0$ matrix to be empty. For user $k$, define the public matrices $C_k\in\F_q^{h\times z}$ and $A_k\in\F_q^{h\times h}$ by
\begin{equation}\label{eq:schur}
 C_k=Q_{\cJ_k}Q_{\cI_k}^{-1},\qquad
 A_k=P_{\cJ_k}-C_kP_{\cI_k}.
\end{equation}
The rows of $C_k$ and $A_k$ retain the labels $T\in\cJ_k$. Let $a_{k,T}$ denote row $T$ of $A_k$. When $z=0$, $C_k$ is the $h\times0$ empty matrix and $C_kP_{\cI_k}$ is the zero matrix. For every $n\in[N-1]$, eliminating the sharing randomness gives
\begin{equation}\label{eq:local-elimination}
 S_{n,\cJ_k}-C_kS_{n,\cI_k}=A_k(W_n-W_N).
\end{equation}
The matrix $A_k$ is invertible. To verify this, let $v\in\F_q^h$ satisfy $A_kv=0$, and set $u=-Q_{\cI_k}^{-1}P_{\cI_k}v\in\F_q^z$. The $\cI_k$ rows of $Pv+Qu$ vanish by the definition of $u$. The $\cJ_k$ rows vanish because $A_kv=0$. Invertibility of $[P\ Q]$ then gives $v=0$.

\subsection{Placement and Delivery}
To hide requests in the delivery coefficients, give user $k$ a private vector $s_k$, uniform on $\F_q^{N-1}$. Write $s_k=(s_{k,1},\ldots,s_{k,N-1})^{\mathsf T}$. Choose these vectors independently across users and independently of the library and sharing randomness. For every $(t+1)$-subset $A$ of $[K]$, also generate a uniform key packet $T_A\in\F_q^L$. The key packets are mutually independent and independent of all the other variables just defined.

To recover the requested file from its difference, each cache must also supply the appropriate reference-file component. For every missing row $T\in\cJ_k$, user $k$ therefore stores a masked correction using the same row $a_{k,T}$ that appears in the decoded difference. We call this packet an \emph{aligned mask}. User $k$ caches its private vector $s_k$, the shares
\begin{equation}
 (S_{n,T}:n\in[N-1],\ T\in\cI_k),
\end{equation}
and, for each $T\in\cJ_k$, the aligned mask
\begin{equation}\label{eq:aligned-mask}
 E_{k,T}=T_{T\cup\{k\}}+\sum_{n=1}^{N-1}s_{k,n}S_{n,T}
              -a_{k,T}W_N.
\end{equation}
The independent key $T_{T\cup\{k\}}$ hides the two data-dependent terms in $E_{k,T}$. During decoding, subtracting $E_{k,T}$ adds $a_{k,T}W_N$ and thus supplies the reference-file correction.

Let $e_1,\ldots,e_{N-1}$ be the standard basis vectors of $\F_q^{N-1}$ and set $e_N=0$. For a fixed demand vector $d\in[N]^K$, define the public control vectors $q_k=(q_{k,1},\ldots,q_{k,N-1})^{\mathsf T}$ by
\begin{equation}\label{eq:controls}
 q_k=s_k+e_{d_k},\qquad k\in[K].
\end{equation}
Write $\boldsymbol q=(q_1,\ldots,q_K)$ for the tuple of public controls. The server broadcasts $\boldsymbol q$ and, for every $(t+1)$-subset $A$ of $[K]$, the coded packet
\begin{equation}\label{eq:delivery}
 Y_A=T_A+\sum_{j\in A}\sum_{n=1}^{N-1}
              q_{j,n}S_{n,A\setminus\{j\}}.
\end{equation}
The full transmission is $X_d=(\boldsymbol q,Y_A:|A|=t+1)$. Since the vectors $s_k$ are independent and uniform, the controls $q_k$ are also independent and uniform for every fixed demand vector.

\subsection{Decoding}
To decode the packet intended for user $k$ in a transmission $Y_A$ with $k\in A$, put $T=A\setminus\{k\}$. For every $j\in A\setminus\{k\}$, the row label $A\setminus\{j\}$ contains $k$. User $k$ therefore has all the shares in the summand indexed by $j$ in~\eqref{eq:delivery}. It subtracts these summands using the public controls, then subtracts its cached mask $E_{k,T}$. The key $T_A$ cancels, leaving
\begin{equation}\label{eq:decoding-first}
 \sum_{n=1}^{N-1}(q_{k,n}-s_{k,n})S_{n,T}+a_{k,T}W_N.
\end{equation}
If $d_k=n<N$, the packet in~\eqref{eq:decoding-first} equals $S_{n,T}+a_{k,T}W_N$. User $k$ collects these packets over all $h$ missing rows and subtracts $C_kS_{n,\cI_k}$. By~\eqref{eq:local-elimination}, it obtains
\begin{equation}\label{eq:decoded-transform}
 A_k(W_n-W_N)+A_kW_N=A_kW_n.
\end{equation}
If $d_k=N$, the share term in~\eqref{eq:decoding-first} is zero. The packets for the $h$ missing rows give $A_kW_N$ directly. For either type of request, user $k$ inverts $A_k$ to recover its file. The cancellation argument applies to each of the $t+1$ users in $A$, so every coded packet serves all those users.

\subsection{A Simulator for the Complete User View}
To prove joint privacy, we construct a simulator that reproduces the user's observation using only its requested index and file. Fix user $k\in[K]$ and its requested index $i\in[N]$. First condition on $D_k=i$ and a fixed value of the public control tuple $\boldsymbol q$. The private vector of user $k$ is then determined as $s_k=q_k-e_i$.

The base observation consists of three groups of packets: the cached shares, the cached aligned masks, and the broadcasts whose labels exclude $k$. Denote their tuple by
\begin{equation}\label{eq:base-view}
 \begin{aligned}
 \mathcal B_k=\bigl(&S_{n,\cI_k}:n<N;\ E_{k,T}:T\in\cJ_k;\\
                         &Y_A:|A|=t+1,\ k\notin A\bigr).
 \end{aligned}
\end{equation}
To determine the distribution of $\mathcal B_k$, fix the full library and demand vector as well. The packets in $(U_n:n<N;\,T_A:|A|=t+1)$ remain independent and uniform under this conditioning, since the controls depend only on the demands and the independent vectors $s_k$.

Order the keys so that those used in the cached masks come first, in the order of those masks. Follow them with the keys used in the broadcasts indexed by sets that exclude $k$, in the order of those broadcasts. For a nonnegative integer $m$, let $I_m$ denote the $m\times m$ identity matrix. Let $\otimes$ denote the Kronecker product, and use $*$ for the block of coefficients of the sharing randomness in the masks and broadcasts. The map from the random packets to $\mathcal B_k$ is affine, with linear part
\begin{equation}\label{eq:base-bijection}
 \begin{bmatrix}
 I_{N-1}\otimes Q_{\cI_k} & 0\\
 * & I_{\binom K{t+1}}
 \end{bmatrix}.
\end{equation}
Each cached mask uses one key whose label contains $k$. Each broadcast included in $\mathcal B_k$ uses a different key whose label excludes $k$. Thus every key occurs in exactly one of the mask or broadcast entries of $\mathcal B_k$. The two diagonal blocks in~\eqref{eq:base-bijection} are invertible, so the full map is invertible and sends independent uniform input packets to independent uniform output packets. The conditional distribution of $\mathcal B_k$ is consequently the same for every library and every demand vector consistent with $D_k=i$ and the fixed controls $\boldsymbol q$.

It remains to reconstruct the broadcasts whose labels contain $k$. They are determined by $\mathcal B_k$, $\boldsymbol q$, and $W_i$. Write $(e_i)_n$ for coordinate $n$ of $e_i$. For $T\in\cJ_k$, the decoding identities give
\begin{align}\label{eq:reconstruct-view}
 Y_{T\cup\{k\}}={}&E_{k,T}+a_{k,T}W_i\\
       &+\left(C_k\sum_{n=1}^{N-1}(e_i)_nS_{n,\cI_k}\right)_T\notag\\
 &+\sum_{j\in T}\sum_{n=1}^{N-1}q_{j,n}
                 S_{n,(T\cup\{k\})\setminus\{j\}}.\notag
\end{align}
The sum inside $C_k$ selects $S_{i,\cI_k}$ when $i<N$ and is zero when $i=N$, since $e_N=0$. Every share on the last line has its row label in $\cI_k$ and hence is included in $\mathcal B_k$.

Given only $(i,W_i)$, the simulator draws independent uniform control vectors to form $\boldsymbol q$, sets $s_k=q_k-e_i$, and draws all packets in $\mathcal B_k$ independently and uniformly. It then generates the remaining broadcasts from~\eqref{eq:reconstruct-view}. For every fixed library $w$ and demand vector $d$ with $d_k=i$, this procedure reproduces the distribution of $(Z_k,X_d)$ using only $(i,w_i)$. It establishes~\eqref{eq:kernel} for the constructed code and hence proves~\eqref{eq:jointprivacy}.

\subsection{Memory, Delivery Rate, and General Cache Sizes}
User $k$ stores $(N-1)z$ share packets and $h$ aligned-mask packets. The server sends $\binom K{t+1}$ coded packets. Dividing by the file size of $h$ packets gives
\begin{align}
 M_t&=\frac{(N-1)z+h}{h}=1+\frac{(N-1)t}{K-t},\\
 R_t&=\frac{\binom K{t+1}}h=\frac K{t+1}.
\end{align}
The private vector in each cache adds $(N-1)\log_2q$ bits. The public control tuple adds $K(N-1)\log_2q$ bits to the transmission. After normalization, these contributions are $(N-1)/(hL)$ and $K(N-1)/(hL)$, respectively, and vanish as $L$ grows. Thus the packet and control counts establish the corner points in Theorem~\ref{thm:achievable} under the closure in the definition of achievability.

To obtain intermediate cache sizes, split each file into independent parts and apply the corner-point schemes with independent placement randomness. Conditional on the full library and demand vector, the observations from the component schemes are independent. Each component distribution depends only on the observing user's requested index and the corresponding part of the requested file. Their product therefore depends only on that requested index and the complete requested file, which preserves joint privacy. Varying the part sizes gives the lower convex envelope $U_{N,K}(M)$.

Write $a=N-1$. The slope between corner points $t$ and $t+1$ is
\begin{equation}\label{eq:s-slope}
 -\frac{(K-t)(K-t-1)}{a(t+1)(t+2)},\qquad 0\le t<K-1.
\end{equation}
The slopes increase with $t$, so the envelope linearly interpolates between consecutive corner points. At cache sizes above the last corner point, the same scheme with unit delivery rate can leave the extra memory unused.

For the constant-factor proof, we also need a simple upper bound that does not require evaluating the convex envelope. Fix $M>1$, and write $\delta=M-1$. Define the real-valued index $t_*$ and its integer rounding $t$ by
\begin{equation}\label{eq:integer-choice}
 t_* =\frac{K\delta}{a+\delta},\qquad t=\lfloor t_*\rfloor.
\end{equation}
Since $t\le t_*<K$, we have $t\in\{0,\ldots,K-1\}$ and $M_t\le M$. Moreover,
\begin{equation}\label{eq:upper-scale}
 U_{N,K}(M)\le\frac K{t+1}
       \le\min\left\{K,\,1+\frac a\delta\right\}.
\end{equation}
The second inequality uses $t+1>t_*$ and $K/t_*=1+a/\delta$. At the endpoint $M=1$, independent one-time-pad unicasts achieve $(1,K)$ exactly.

\subsection{A Three-User Example}
To illustrate the reference-file correction and the privacy simulator, take $N=K=3$ and $t=1$, so $(F,z,h)=(3,1,2)$. Arithmetic is over $\F_5$ to keep the coefficients small. Each file is a pair of independent uniform packets $W_n=(W_{n,1},W_{n,2})^{\mathsf T}$, with each packet in $\F_5^L$. In this example, use the index $j$ as shorthand for the singleton row label $\{j\}$. Let $x_1=1,x_2=2,x_3=3$ be the three fixed evaluation points. For $n=1,2$, draw a uniform packet $U_n\in\F_5^L$, independently of the library and of the other packet, and form
\begin{equation}\label{eq:s-example-shares}
 S_{n,j}=U_n+x_j(W_{n,1}-W_{3,1})
                   +x_j^2(W_{n,2}-W_{3,2}).
\end{equation}
Write $\boldsymbol S_j=(S_{1,j},S_{2,j})^{\mathsf T}$ for the two shares with row index $j$. User $k$ caches $\boldsymbol S_k$. Conditional on any library, the two-packet vector $\boldsymbol S_k$ is uniform because $(U_1,U_2)$ is uniform. A user's cached shares alone therefore reveal no information about the library.

Under the singleton-label shorthand, the cached row set is $\cI_k=\{k\}$ and the missing rows are the other two indices in increasing order. For each user, $C_k=(1,1)^{\mathsf T}$.

\begin{remark}[Local transforms over $\F_5$]\label{rem:example-transforms}
The rows of $A_1$, $A_2$, and $A_3$ correspond to the missing indices $(2,3)$, $(1,3)$, and $(1,2)$, respectively. The transforms and their determinants are
\begin{align*}
 A_1&=\begin{bmatrix}1&3\\2&3\end{bmatrix},&\det A_1&=2,\\
 A_2&=\begin{bmatrix}4&2\\1&0\end{bmatrix},&\det A_2&=3,\\
 A_3&=\begin{bmatrix}3&2\\4&0\end{bmatrix},&\det A_3&=2.
\end{align*}
For instance, subtracting the coefficient row at $x_1$ from those at $x_2$ and $x_3$ gives the rows $(1,3)$ and $(2,3)$ of $A_1$.
\end{remark}

For compact notation, write $T_{ab}=T_{\{a,b\}}$ and $Y_{ab}=Y_{\{a,b\}}$ when $a<b$. Draw independent uniform key packets $T_{12},T_{13},T_{23}\in\F_5^L$ and independent uniform vectors $s_k\in\F_5^2$, all independently of the library and sharing randomness. Let $a_{k,j}$ be the row of $A_k$ with missing index $j$, and define
\begin{equation}\label{eq:s-example-cache}
 \begin{aligned}
 E_{k,j}&=T_{\{k,j\}}+s_k^{\mathsf T}\boldsymbol S_j-a_{k,j}W_3,\\
 Z_k&=(s_k,\boldsymbol S_k,E_{k,j}:j\ne k).
 \end{aligned}
\end{equation}
Set $e_1=(1,0)^{\mathsf T}$, $e_2=(0,1)^{\mathsf T}$, and $e_3=0$. For demand vector $d$, the server publishes the controls $q_k=s_k+e_{d_k}$ and sends
\begin{align}\label{eq:s-example-delivery}
 Y_{12}&=T_{12}+q_1^{\mathsf T}\boldsymbol S_2+q_2^{\mathsf T}\boldsymbol S_1,\\
 Y_{13}&=T_{13}+q_1^{\mathsf T}\boldsymbol S_3+q_3^{\mathsf T}\boldsymbol S_1,\\
 Y_{23}&=T_{23}+q_2^{\mathsf T}\boldsymbol S_3+q_3^{\mathsf T}\boldsymbol S_2.
\end{align}
Each packet serves its two indexed users. For user $k$ and $j\ne k$, cancellation gives
\begin{equation}\label{eq:s-example-cancellation}
 Y_{\{k,j\}}-q_j^{\mathsf T}\boldsymbol S_k-E_{k,j}
       =e_{d_k}^{\mathsf T}\boldsymbol S_j+a_{k,j}W_3.
\end{equation}
If $d_k=n<3$, subtract the cached packet $S_{n,k}$ from the right side. Equation~\eqref{eq:s-example-shares} then gives $a_{k,j}W_n$. If $d_k=3$, the share term is zero and the right side already equals $a_{k,j}W_3$. Collecting the packets for the two missing indices gives $A_kW_{d_k}$. User $k$ recovers its file by inverting $A_k$, whose determinant is nonzero by Remark~\ref{rem:example-transforms}.

For user 1, condition on the public control tuple $\boldsymbol q$ and $D_1=i$. The five packets in
\begin{equation}\label{eq:s-example-base}
 (\boldsymbol S_1,E_{1,2},E_{1,3},Y_{23})
\end{equation}
are jointly uniform, with the same distribution for every library and every pair of other-user demands. The random packets $U_1,U_2$ make the two packets in $\boldsymbol S_1$ uniform. Conditional on those packets and any fixed library and demand vector, the independent keys $T_{12},T_{13},T_{23}$ make $E_{1,2},E_{1,3},Y_{23}$ independent and uniform. The broadcasts $Y_{12}$ and $Y_{13}$ are determined by the tuple in~\eqref{eq:s-example-base}, $W_i$, and $\boldsymbol q$ through the rearranged decoding equations. Drawing the controls and the five base packets and then reconstructing those two broadcasts gives the simulator for user 1. Relabeling the users gives the simulators for users 2 and 3.

Each cache contains four packets, each of size $B/2$ bits, and two scalar control symbols. Delivery contains three packets and six control symbols. Since each file contains two packets, the finite-packet counts are
\begin{equation}
 M=2+\frac1L,\qquad R=\frac32+\frac3L.
\end{equation}
As $L$ grows, these approach $(M,R)=(2,3/2)$.

%% file: sections/04_converse.tex
\section{A Converse From User Groups and Alternative Demands}\label{sec:converse}
To derive a lower bound that accounts for both the file count and the user count, we compare several demand vectors under a fixed placement. One user, the anchor, always requests $W_1$. In each demand vector, one group of other users requests files from $W_2,\ldots,W_N$, while everyone outside the group requests $W_1$. The anchor's privacy keeps the group's requested files hidden from the anchor's observation. We then add the groups in sequence. Under each demand assigned to the current group, users in earlier groups request $W_1$. Their cache entropies can therefore be bounded after conditioning on $W_1$ and a broadcast for the current group.

All the random variables $X_d$ below are defined on a common probability space using the same library and placement randomness. The different $X_d$ represent alternative deliveries for one placement. The privacy conditions are applied separately to each delivery.

\subsection{General Group and Demand Profiles}
Set $a=N-1$ and $\delta=M-1$. Fix user 1 as the anchor. Choose a positive integer $g$, disjoint user groups $\cG_1,\ldots,\cG_g\subseteq[K]\setminus\{1\}$, and positive integers $\ell_1,\ldots,\ell_g$. The groups and their demand vectors are fixed choices. For each $j\in[g]$, let $r_j=|\cG_j|$ and require
\begin{equation}\label{eq:profile-constraints}
 r_j\ell_j\le a,\qquad \sum_{j=1}^g r_j\le K-1.
\end{equation}
For group $j$, use $\ell_j$ alternative demand vectors. Across the $\ell_j$ vectors, the group requests $r_j\ell_j$ different files from $W_2,\ldots,W_N$. Every user outside $\cG_j$ requests $W_1$. File indices may be reused in different groups. The group sizes and demand counts determine a profile. Define its coefficients by
\begin{align}\label{eq:profile-coefficients}
 \alpha&=\sum_{j=1}^g r_j\ell_j,\qquad
 \beta=\sum_{j=1}^g(g-j+1)r_j,\\
 \gamma&=1+\sum_{j=1}^g(\ell_j-1).\notag
\end{align}

\begin{theorem}[Grouped anchored bound]\label{thm:grouped}
Every jointly private code satisfies, for each profile obeying~\eqref{eq:profile-constraints},
\begin{equation}\label{eq:general-converse}
 R\ge1+\frac{\alpha-\beta(M-1)}\gamma.
\end{equation}
The bound holds for arbitrary randomized placement and delivery. It also holds for any correct code with fixed-demand file privacy $I(W_{[N]\setminus\{d_k\}};Z_k,X_d)=0$ for every user $k$ and demand vector $d$. External secrecy $I(W_{[N]};X_d)=0$ is also required for every $d$.
\end{theorem}

The coefficient $\alpha$ sums the number of distinct files requested within each group. A file may be counted again in another group. The coefficient $\beta$ counts how often cache entropy is subtracted as groups are added. The coefficient $\gamma$ counts the broadcast terms that remain after the lower and upper bounds on joint cache entropy are compared.

\subsection{Proof of Theorem~\ref{thm:grouped}}
Let $\varepsilon=R-1$. Denote by $X_0$ the broadcast when every user requests $W_1$. For each $j\in[g]$ and $l\in[\ell_j]$, let $X_{j,l}$ be the broadcast under the $l$-th demand vector assigned to group $j$. Collect these broadcasts in $V_j=(X_{j,1},\ldots,X_{j,\ell_j})$. Let $\mathcal A_j\subseteq\{2,\ldots,N\}$ contain the $r_j\ell_j$ file indices requested by group $j$ across its demand vectors.

\subsubsection{The information supplied by one group}
Under the first demand assigned to group $j$, the anchor requests $W_1$. File privacy keeps $W_{\mathcal A_j}$ independent of the anchor's view. Correctness allows us to condition additionally on $W_1$, because the anchor decodes that file from its view. Hence
\begin{equation}\label{eq:anchor-file-entropy}
 H(W_{\mathcal A_j}\mid Z_1,W_1,X_{j,1})=r_j\ell_j B.
\end{equation}
For $l\in\{2,\ldots,\ell_j\}$, each additional broadcast satisfies the bound in~\eqref{eq:delivery-budget}:
\begin{equation}
 H(X_{j,l}\mid Z_1,W_1)\le\varepsilon B.
\end{equation}
Adding the remaining $\ell_j-1$ broadcasts to the conditioning in~\eqref{eq:anchor-file-entropy} can therefore reduce the conditional entropy by at most $(\ell_j-1)\varepsilon B$ bits. Thus
\begin{equation}\label{eq:group-files}
 H(W_{\mathcal A_j}\mid Z_1,W_1,V_j)
 \ge\bigl[r_j\ell_j-(\ell_j-1)\varepsilon\bigr]B.
\end{equation}
By correctness, the cache tuple $Z_{\cG_j}$ and the broadcasts $V_j$ together determine $W_{\mathcal A_j}$. Consequently,
\begin{equation}\label{eq:group-cache}
 H(Z_{\cG_j}\mid Z_1,W_1,V_j)
 \ge\bigl[r_j\ell_j-(\ell_j-1)\varepsilon\bigr]B.
\end{equation}
This step uses correctness for several users and privacy only for the single anchor.

\subsubsection{Adding groups in sequence}
Let $\cG_{<j}=\bigcup_{i<j}\cG_i$ be the set of users in earlier groups, and let $n_{j-1}=\sum_{i<j}r_i$ be their number. For $j=1$, the set is empty and $n_0=0$. All users in $\cG_{<j}$ request $W_1$ under every demand assigned to group $j$. Since $V_j$ contains $X_{j,1}$, conditional subadditivity and~\eqref{eq:cache-budget} give
\begin{equation}\label{eq:previous-caches}
 H(Z_{\cG_{<j}}\mid Z_1,W_1,V_j)\le n_{j-1}\delta B.
\end{equation}
To replace $V_j$ in the conditioning by the earlier caches while bounding the entropy loss, we use the following identity for arbitrary random variables $A,B',C,V$:
\begin{align*}
&H(A\mid B',C)-H(A\mid C,V)+H(B'\mid C,V)\\
&\qquad=I(A;V\mid B',C)+H(B'\mid A,C,V)\ge0.
\end{align*}
Apply the identity with $A=Z_{\cG_j}$, $B'=Z_{\cG_{<j}}$, $C=(W_1,Z_1)$, and $V=V_j$. Together with~\eqref{eq:group-cache} and~\eqref{eq:previous-caches}, it yields
\begin{align}
 &H(Z_{\cG_j}\mid W_1,Z_1,Z_{\cG_{<j}})\notag\\
 &\quad\ge H(Z_{\cG_j}\mid W_1,Z_1,V_j)
              -H(Z_{\cG_{<j}}\mid W_1,Z_1,V_j)\notag\\
 &\quad\ge\bigl[r_j\ell_j-(\ell_j-1)\varepsilon
                       -n_{j-1}\delta\bigr]B.\label{eq:group-increment}
\end{align}
Summing~\eqref{eq:group-increment} over the groups gives a lower bound on the joint entropy of the anchor's cache and the group caches. Since the anchor's cache is independent of the library, $H(Z_1\mid W_1)=H(Z_1)$, and
\begin{align}\label{eq:cache-lower}
 &H(Z_1,Z_{\cG_1},\ldots,Z_{\cG_g}\mid W_1)\notag\\
 &\quad\ge H(Z_1)+\left[\alpha-(\gamma-1)\varepsilon
                          -\delta\sum_{j=1}^g n_{j-1}\right]B.
\end{align}
Larger groups contribute several independent files in one increment. The subtraction in~\eqref{eq:group-increment} depends on the number of users in earlier groups.

\subsubsection{An upper bound on the same cache entropy}
We now bound the same joint cache entropy from above using the demand in which every user requests $W_1$. Let $n=\sum_{j=1}^g r_j$ be the total number of users in the groups. For the anchor,~\eqref{eq:cache-exact} gives
\begin{equation}\label{eq:anchor-exact-cache}
 H(Z_1\mid W_1,X_0)\le H(Z_1)-B.
\end{equation}
After conditioning on $W_1$ and $X_0$, each of the other $n$ caches has at most $\delta B$ bits of entropy by~\eqref{eq:cache-budget}. Therefore
\begin{align}\label{eq:cache-upper}
 &H(Z_1,Z_{\cG_1},\ldots,Z_{\cG_g}\mid W_1)\notag\\
 &\quad\le H(X_0\mid W_1)+H(Z_1\mid W_1,X_0)+n\delta B\notag\\
 &\quad\le H(Z_1)+(\varepsilon+n\delta)B.
\end{align}
Comparing~\eqref{eq:cache-lower} and~\eqref{eq:cache-upper} cancels the actual entropy $H(Z_1)$. Since
\begin{equation}
 n+\sum_{j=1}^g n_{j-1}=\sum_{j=1}^g(g-j+1)r_j=\beta,
\end{equation}
we obtain $\gamma\varepsilon\ge\alpha-\beta\delta$, proving~\eqref{eq:general-converse}. The proof uses only fixed-demand file privacy, external library secrecy, and correctness. It therefore also establishes the final statement of the theorem.\hfill$\square$

\subsection{Many Groups, One Demand per Group}
We first obtain an explicit bound by assigning one demand vector to each group. To express the cache coefficient for this choice, define, for every nonnegative integer $n$,
\begin{equation}\label{eq:J-definition}
 J_a(n)=\sum_{i=1}^n\left\lceil\frac ia\right\rceil,
 \qquad J_a(0)=0.
\end{equation}
Write $n=ab+r$, where $b=\lfloor n/a\rfloor$ is a nonnegative integer and $r\in\{0,\ldots,a-1\}$. Then
\begin{equation}\label{eq:J-closed}
 J_a(n)=\frac{ab(b+1)}2+(b+1)r.
\end{equation}
For each integer $n\in\{1,\ldots,K-1\}$, select $n$ users other than the anchor and set every $\ell_j=1$. If $r>0$, place the first $r$ users in an initial group. Partition the remaining $ab$ users into $b$ groups of size $a$. The resulting profile has $\alpha=n$, $\gamma=1$, and $\beta=(b+1)r+a\sum_{i=1}^{b}i=J_a(n)$. Theorem~\ref{thm:grouped}, together with $R\ge1$ for the $n=0$ term, yields
\begin{equation}\label{eq:group-bound}
 \LB(M)=\max_{0\le n\le K-1}\{1+n-\delta J_a(n)\}
             \le\Rfd(M).
\end{equation}
The maximum is over integer values of $n$. For $\delta>0$, one maximizing choice is
\begin{equation}\label{eq:n-optimal}
 n=\min\{K-1,\ a\lfloor1/\delta\rfloor\}.
\end{equation}
Indeed, increasing $n$ from $i-1$ to $i$ changes the objective by $1-\delta\lceil i/a\rceil$. The increments are nonincreasing. The choice in~\eqref{eq:n-optimal} includes all nonnegative increments allowed by the user count.

The summands in~\eqref{eq:J-definition} consist of $a$ occurrences of 1, followed by $a$ occurrences of 2, and so on. This pattern reflects the fact that one group can include up to $a=N-1$ users requesting different files. For example, with $N=4$ and $K=8$, groups of sizes $1,3,3$ give $J_3(7)=12$ and
\begin{equation}
 R\ge8-12(M-1).
\end{equation}
Each three-user group requests the three files other than $W_1$. Those files can be reused for the next group. Under the next group's demand, the earlier users request $W_1$, so their caches have bounded entropy after conditioning on $W_1$ and the corresponding broadcast.

\subsection{One Group, Several Demands}
To obtain a complementary bound for larger values of $M-1$, we let one group request several batches of files. Set $s_{\max}=\min\{a,K-1\}$. For each integer $s\in\{1,\ldots,s_{\max}\}$, choose $g=1$, $r_1=s$, and $\ell_1=\lfloor a/s\rfloor$. Each batch contains $s$ different files, and the batches are disjoint. The choice of $\ell_1$ uses as many batches as the $a$ available files permit. For a real number $x$, write $[x]^+=\max\{x,0\}$. Define
\begin{equation}\label{eq:multidemand-bound}
 \LA(M)=1+\max_{1\le s\le s_{\max}}
 \left[s\left(1-\frac{\delta}{\lfloor a/s\rfloor}\right)\right]^+.
\end{equation}
By~\eqref{eq:general-converse} and $R\ge1$, we have $\LA(M)\le\Rfd(M)$. To compare the achievable delivery rate with both bounds, define
\begin{equation}\label{eq:L-definition}
 L(M)=\max\{\LB(M),\LA(M)\}.
\end{equation}

%% file: sections/05_gap.tex
\section{\texorpdfstring{A Uniform Factor-$9/2$ Guarantee}{A Uniform Factor-9/2 Guarantee}}\label{sec:gap}
To prove a uniform bound on the ratio between the achievable and optimal delivery rates, we compare both with a common scale. Define
\begin{equation}\label{eq:Phi}
 \Phi_{N,K}(M)=
 \begin{cases}
 K,&M=1,\\[1mm]
 \displaystyle\min\left\{K,1+\frac{N-1}{M-1}\right\},&M>1.
 \end{cases}
\end{equation}
The term $K$ is the delivery rate of independently encrypted unicast. For $M>1$, the term $1+(N-1)/(M-1)$ bounds the delivery rate obtained by the integer choice in~\eqref{eq:integer-choice}.

\begin{theorem}[Order optimality for all parameters]\label{thm:gap}
For all integers $N,K\ge2$ and every $M\ge1$,
\begin{equation}\label{eq:main-chain}
 \begin{aligned}
 \frac29\Phi_{N,K}(M)&\le L(M)\le\Rfd(M),\\
 \Rfd(M)&\le U_{N,K}(M)\le\Phi_{N,K}(M).
 \end{aligned}
\end{equation}
In particular,
\begin{equation}\label{eq:gap}
 1\le\frac{U_{N,K}(M)}{\Rfd(M)}\le\frac92.
\end{equation}
All constants are independent of $N$, $K$, and $M$.
\end{theorem}

\subsection{Proof of Theorem~\ref{thm:gap}}
The inequalities $\Rfd(M)\le U_{N,K}(M)\le\Phi_{N,K}(M)$ follow from Theorem~\ref{thm:achievable} and~\eqref{eq:upper-scale}. The bound $L(M)\le\Rfd(M)$ follows from~\eqref{eq:group-bound} and~\eqref{eq:multidemand-bound}. It remains to prove $L(M)\ge(2/9)\Phi_{N,K}(M)$. Write $a=N-1$ and let $\delta=M-1$ denote the effective cache memory. For $\delta=0$, choosing $n=K-1$ in~\eqref{eq:group-bound} gives $L(M)=K$. Assume henceforth that $\delta>0$.

\subsubsection{\texorpdfstring{Small effective caches: $0<\delta\le2/3$}{Small effective caches: 0 < delta <= 2/3}}
For small effective cache memory, we use the grouped bound with the maximizing user count in~\eqref{eq:n-optimal}. Define the real number $x$ and the positive integers $g,n$ by
\begin{equation}
 x=1/\delta,\qquad g=\lfloor x\rfloor,\qquad
 n=\min\{K-1,ag\}.
\end{equation}
If $K-1\le ag$, then $n=K-1$. Because the summands in~\eqref{eq:J-definition} are nondecreasing, the average of the first $n$ summands is at most the average of the first $ag$ summands. Hence $J_a(n)/n\le(g+1)/2$, and
\begin{equation}\label{eq:gap-small-capped}
 \LB(M)\ge1+n\left(1-\frac{\delta(g+1)}2\right).
\end{equation}
For $g=1$, the factor $1-\delta(g+1)/2$ is at least $1/3$. For $g\ge2$, the inequality $\delta\le1/g$ bounds that factor below by $(g-1)/(2g)\ge1/4$. It follows that
\begin{equation}
 \LB(M)\ge K/4\ge(2/9)\Phi_{N,K}(M).
\end{equation}

If $K-1>ag$, the selected user count is $n=ag$. Substituting this value in~\eqref{eq:group-bound} gives
\begin{equation}\label{eq:gap-small-uncapped}
 \LB(M)\ge1+a\left[g-\frac{\delta g(g+1)}2\right].
\end{equation}
For $\delta\le1/2$, let $r=x-g$ be the fractional part of $x$, so $r\in[0,1)$. Since $x\ge2$, rearranging gives
\begin{align}
 g-\frac{g(g+1)}{2x}
 &=\frac{x-1}{2}+\frac{r(1-r)}{2x}\notag\\
 &\ge\frac x4\ge\frac{2x}{9}.\label{eq:rounding-small}
\end{align}
For $1/2<\delta\le2/3$, we have $g=1$ and $\delta(1-\delta)\ge2/9$. Thus the expression in square brackets in~\eqref{eq:gap-small-uncapped} is again at least $2/(9\delta)$. The two ranges of $\delta$ therefore give
\begin{equation}\label{eq:gap-small-final}
 \LB(M)\ge1+\frac{2a}{9\delta}
      \ge\frac29\left(1+\frac a\delta\right)
      \ge\frac29\Phi_{N,K}(M).
\end{equation}

\subsubsection{\texorpdfstring{Larger effective caches: $\delta\ge2/3$}{Larger effective caches: delta >= 2/3}}
For larger effective cache memory, we choose the demand count $\ell$ in the multiple-demand bound so that $z_0=\delta/\ell$ lies in $[1/3,1/2]$. We then take the largest group size $s$ allowed by the file and user counts:
\begin{equation}\label{eq:rounding-large}
 \ell=\lceil2\delta\rceil,\qquad z_0=\frac\delta\ell,\qquad
 s=\min\{K-1,\lfloor a/\ell\rfloor\}.
\end{equation}
The choice of $\ell$ gives
\begin{equation}\label{eq:z-interval}
 \frac13\le z_0\le\frac12.
\end{equation}
The upper bound on $z_0$ follows from $\ell\ge2\delta$. For $\delta\in[2/3,1]$, we have $\ell=2\le3\delta$. For $\delta\ge1$, we have $\ell\le2\delta+1\le3\delta$. Thus $\ell\le3\delta$ throughout the current range, which gives the lower bound on $z_0$.

When $s\ge1$, define $\ell_s=\lfloor a/s\rfloor$. Since $s\ell\le a$, we have $\ell_s\ge\ell$. Evaluating the maximization in~\eqref{eq:multidemand-bound} at the selected group size $s$ gives $\LA(M)\ge1+s(1-\delta/\ell_s)\ge1+s(1-\delta/\ell)$. Therefore
\begin{equation}\label{eq:large-selected}
 \LA(M)\ge1+s(1-z_0).
\end{equation}
Bound~\eqref{eq:large-selected} also holds for $s=0$ because $\LA(M)\ge1$. If $s=K-1$, then $z_0\le1/2$ yields $\LA(M)\ge(K+1)/2\ge\Phi_{N,K}(M)/2$.

If $s<K-1$, then $s=\lfloor a/\ell\rfloor\ge a/\ell-1$. Substituting this lower bound on $s$ into~\eqref{eq:large-selected} gives
\begin{align}
 \LA(M)&\ge z_0+\frac a\delta z_0(1-z_0)\notag\\
       &\ge\frac29\left(1+\frac a\delta\right)
        \ge\frac29\Phi_{N,K}(M).\label{eq:large-final}
\end{align}
The second inequality uses $z_0\ge1/3$ and $z_0(1-z_0)\ge2/9$ for $z_0\in[1/3,1/2]$. The two effective-memory ranges cover every $\delta>0$, proving~\eqref{eq:main-chain} and~\eqref{eq:gap}.\hfill$\square$

\subsection{Scaling and the Role of Demand Privacy}
Theorem~\ref{thm:gap} shows that the optimal delivery rate scales as
\begin{equation}\label{eq:theta}
 \Rfd(M)=\Theta\!\left(\min\left\{K,1+\frac{N-1}{M-1}\right\}\right)
\end{equation}
for $M>1$, with $\Rfd(1)=K$. Three regimes follow for the effective cache memory $M-1$. Memory of order $(N-1)/K$ or smaller gives a delivery rate of order $K$. Between the memory scales $(N-1)/K$ and $N-1$, the delivery rate is of order $(N-1)/(M-1)$. Effective memory of order $N-1$ or larger gives a delivery rate bounded by a constant independent of all parameters. Section~\ref{sec:boundaries} treats the case of unit delivery rate.

To quantify the additional cost of demand privacy, consider the placement and decoding model of Section~\ref{sec:model} with only file privacy and external secrecy imposed. For every fixed demand vector $d$, require $I(W_{[N]\setminus\{d_k\}};Z_k,X_d)=0$ for every user $k$. Also require external secrecy, $I(W_{[N]};X_d)=0$, for every $d$. File privacy here protects the entire tuple of nonrequested files. Let $\Rfe(M)$ denote the infimum of achievable delivery rates in this model, with the same closure convention as $\Rfd(M)$. Theorem~\ref{thm:grouped} applies to this model, so
\begin{equation}
 L(M)\le\Rfe(M)\le\Rfd(M)\le U_{N,K}(M).
\end{equation}
Combining this chain with $U_{N,K}(M)\le\Phi_{N,K}(M)\le(9/2)L(M)$ from Theorem~\ref{thm:gap} gives
\begin{equation}\label{eq:privacy-price}
 \Rfd(M)\le\frac92\Rfe(M).
\end{equation}
Thus adding joint demand privacy preserves the order of the optimal memory--rate tradeoff when file privacy and external secrecy are already required.

\subsection{Sharpness of the Two-Bound Comparison}
We now show that the factor $9/2$ cannot be reduced in a uniform comparison of the achievable envelope $U_{N,K}(M)$ with the explicit lower bound $L(M)=\max\{\LB(M),\LA(M)\}$. Let $a$ tend to infinity through the positive integers, and consider
\begin{equation}\label{eq:sharp-sequence}
 N=a+1,\quad K=a^2+2,\quad M=5/3,\quad a\longrightarrow\infty.
\end{equation}
Along this sequence,~\eqref{eq:n-optimal} gives $\LB(5/3)=1+a/3$. To evaluate~\eqref{eq:multidemand-bound}, fix an integer $s\in\{1,\ldots,a\}$ and set $\ell=\lfloor a/s\rfloor$. Then
\begin{equation}
 s\left(1-\frac{2}{3\ell}\right)
 \le a\left(\frac1\ell-\frac{2}{3\ell^2}\right)\le\frac a3.
\end{equation}
The last inequality is equivalent to $(\ell-1)(\ell-2)\ge0$, which holds because $\ell$ is a positive integer. Equality is attained by $s=a$, so $\LA(5/3)=1+a/3$ as well.

For the achievable envelope, let $t_*=K(M-1)/(a+M-1)$ as in~\eqref{eq:integer-choice}. Set $t_-=\lfloor t_*\rfloor$ and $t_+=t_-+1$. These are the indices of adjacent achievable corner points whose memory coordinates bracket $M=5/3$. Both indices satisfy $t_\pm\sim2a/3$. Their delivery rates therefore satisfy $R_{t_-}\sim3a/2$ and $R_{t_+}\sim3a/2$. Since the envelope linearly interpolates between these two points, $U_{a+1,a^2+2}(5/3)\sim3a/2$. It follows that
\begin{equation}\label{eq:sharp-ratio}
 \frac{U_{a+1,a^2+2}(5/3)}{\max\{\LB(5/3),\LA(5/3)\}}
       \longrightarrow\frac92.
\end{equation}

%% file: sections/06_boundaries.tex
\section{Exact Boundary Regimes}\label{sec:boundaries}
The grouped converse determines the optimal delivery rate at $M=1$. To study $R$ near one, we first lower-bound the joint entropy of a family of broadcasts using the files decoded from those broadcasts and the users' caches. We then upper-bound the same entropy using a single user's cache.

\begin{theorem}[Complementary anchored bound]\label{thm:dual}
For every $s\in\{1,\ldots,K\}$, any feasible pair satisfies
\begin{equation}\label{eq:dual-bound}
 M\ge1+(N-1)(s-1)-(N-1)\binom s2(R-1).
\end{equation}
\end{theorem}

The proof in Appendix~\ref{app:dual} applies the same residual entropy bounds as Theorem~\ref{thm:grouped}. It selects $s$ users and takes user 1 as the anchor. For each of the other $s-1$ users, it considers $N-1$ alternative demands, one for each file other than $W_1$. The entropy of the broadcast in which every user requests $W_1$ appears in both bounds on the joint broadcast entropy and cancels. The remaining terms give the coefficient $\binom s2$.

With $a=N-1$ and $\delta=M-1$, Theorem~\ref{thm:dual} gives the explicit lower bound on the delivery rate
\begin{equation}\label{eq:dual-load}
 \LD(M)=1+\max_{2\le s\le K}
 \left[\frac{a(s-1)-\delta}{a\binom s2}\right]^+.
\end{equation}

\begin{corollary}[Endpoints and an exact general segment]\label{cor:endpoints}
For all $N,K\ge2$,
\begin{equation}\label{eq:exact-endpoints}
 \Rfd(1)=K,\qquad \Mone=1+(N-1)(K-1).
\end{equation}
Let $M_{\max}=1+(N-1)(K-1)$. Throughout
\begin{equation}\label{eq:last-range}
 1+\frac{(N-1)(K-2)}2\le M\le M_{\max},
\end{equation}
the exact tradeoff is
\begin{equation}\label{eq:last-segment}
 \Rfd(M)=1+\frac{M_{\max}-M}{(N-1)\binom K2}.
\end{equation}
For $M\ge M_{\max}$, the optimum equals $1$.
\end{corollary}
\begin{proof}
At $M=1$,~\eqref{eq:group-bound} with $n=K-1$ gives $R\ge K$. Independently encrypted unicast attains this delivery rate. At $R=1$, Theorem~\ref{thm:dual} with $s=K$ gives $M\ge M_{\max}$. The construction with $t=K-1$ in Theorem~\ref{thm:achievable} attains this memory. The straight line through the achievable points with $t=K-2$ and $t=K-1$ is~\eqref{eq:last-segment}. This line is also the lower bound on the delivery rate obtained from~\eqref{eq:dual-bound} with $s=K$. Memory sharing therefore attains the lower bound over~\eqref{eq:last-range}. For $M\ge M_{\max}$, the construction at unit delivery rate and the lower bound $R\ge1$ prove that the optimum is one.
\end{proof}

Our construction at unit delivery rate serves every demand vector with the memory $M_{\max}$ in~\eqref{eq:exact-endpoints}. This memory is $K-1$ file units smaller than that of the general point at unit delivery rate reported by Qi--Ravi~\cite{qi2022}. When $N\ge K$, it is also $K-2$ file units smaller than that of the distinct-demand point at unit delivery rate reported by Qi \emph{et al.}~\cite{qi2024}. Table~\ref{tab:positioning} lists the demand assumptions for these comparisons. Under file privacy alone, the exact minimum memory at unit delivery rate in~\cite{fang2025} is $(N-1)(K-1)$. Thus joint file and demand privacy requires exactly one additional file unit at unit delivery rate.

The exact segment covers the whole nonconstant tradeoff when $K=2$. In this case,
\begin{equation}\label{eq:two-users}
 \Rfd(M)=\max\left\{1,\frac{2N-1-M}{N-1}\right\},\qquad M\ge1.
\end{equation}
For $N=2$, the grouped bound also matches the first segment of the construction:
\begin{equation}\label{eq:two-files}
 \Rfd(M)=K-\binom K2(M-1),\quad
 1\le M\le\frac K{K-1}.
\end{equation}
For the lower bound in~\eqref{eq:two-files}, use $a=1$ and $n=K-1$ in~\eqref{eq:group-bound}. Memory sharing between the achievable points with $t=0$ and $t=1$ attains this bound. For $N=K=2$,~\eqref{eq:two-users} recovers the exact region in~\cite[Th.~2]{qi2022}.

%% file: sections/07_comparisons.tex
\section{Numerical Results and Comparisons}\label{sec:numerics}
We compare the achievable convex envelopes with the explicit converse bounds. Define the displayed anchored lower bound by
\begin{equation}\label{eq:display-lower}
 \Lall(M)=\max\{\LB(M),\LA(M),\LD(M)\}.
\end{equation}
For comparison, we also evaluate the classical file-private cut-set bound~\cite[Th.~2]{private2018}. Set $s_c=\min\{\lfloor N/2\rfloor,K\}$. For each integer $s\in\{1,\ldots,s_c\}$, let $b_s=\lfloor N/s\rfloor$. The bound is
\begin{equation}\label{eq:classic-cutset}
 L_{\mathrm{cut}}(M)=1+\max_{1\le s\le s_c}
 \frac{(s-1)[b_s-M]^+}{b_s-1}.
\end{equation}
Joint privacy includes file privacy, so $L_{\mathrm{cut}}(M)$ also lower-bounds the optimal delivery rate in our model.

\subsection{Memory Savings and an Exact Large-Cache Segment}
To illustrate the memory savings and the exact segment near unit delivery rate, Figure~\ref{fig:tradeoff} considers $N=K=4$. The difference-sharing construction has corner points
\begin{equation}
 (1,4),\quad (2,2),\quad (4,4/3),\quad (10,1).
\end{equation}
To compare the schemes over the full demand space, we use the secret-sharing/private-key (SPK) implementation given in Appendix~\ref{app:spk}. It secret-shares all $N$ files and uses independent uniform privacy keys. This implementation achieves the SPK numerical family in~\cite[Th.~1]{qi2024} for every demand vector. After including encrypted unicast, its delivery rates $4$, $2$, $4/3$, and $1$ require memories $1$, $7/3$, $5$, and $13$, respectively. We also display memory sharing between the two general endpoints of Qi--Ravi~\cite[Th.~3]{qi2022}.

\begin{figure}[t]
\centering
\includegraphics[width=\columnwidth]{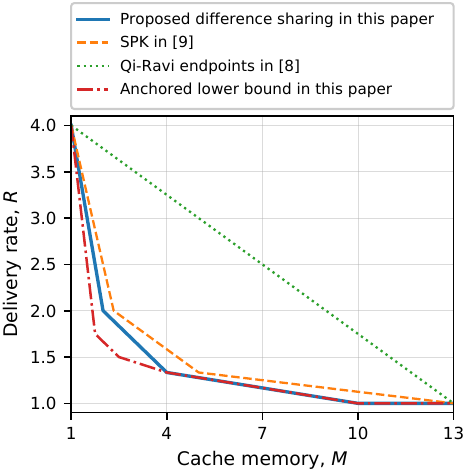}
\caption{All-demand memory--rate comparison for $N=K=4$. The anchored lower bound is~\eqref{eq:display-lower}. The achieved and lower-bound curves agree at $M=1$ and throughout $4\le M\le10$.}
\label{fig:tradeoff}
\end{figure}

At delivery rate $2$, difference sharing requires memory $2$ compared with the SPK value $7/3$. At unit delivery rate, the required memories are $10$ and $13$, respectively. By Corollary~\ref{cor:endpoints}, difference sharing attains the optimal delivery rate throughout $4\le M\le10$.

\subsection{Many Users and Small Caches}
To examine small caches when there are many more users than files, Figure~\ref{fig:smallcache} uses $N=5$ and $K=50$. At $M=1$, the grouped bound gives $50$, matching the achievable delivery rate. At $M=6/5$, the achievable delivery rate and three lower bounds are
\begin{equation}
 \begin{aligned}
 U_{5,50}(6/5)&=151/10,\quad \LB(6/5)=9,\\
 \LA(6/5)&=21/5,\quad L_{\mathrm{cut}}(6/5)=9/5.
 \end{aligned}
\end{equation}
At this memory, the grouped bound is five times the classical cut-set value. At larger cache sizes, the multi-demand bound can exceed the grouped bound. The maximum of the grouped and multi-demand bounds gives the uniform guarantee of Theorem~\ref{thm:gap}.

\begin{figure}[t]
\centering
\includegraphics[width=\columnwidth]{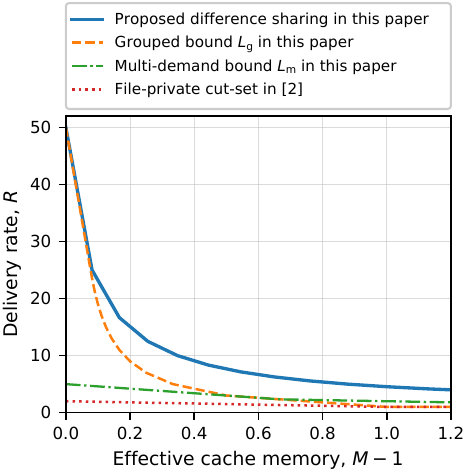}
\caption{Small-cache behavior for $N=5$ and $K=50$. The grouped bound matches the delivery rate at $M=1$ and improves on the displayed cut-set bound for small caches. The multi-demand bound is useful at larger cache sizes.}
\label{fig:smallcache}
\end{figure}

%% file: sections/08_appendices.tex
\section{Proof of the Private-View Representation}\label{app:views}
Fix a user $k\in[K]$ and a requested index $i\in[N]$. Condition on $D_k=i$, and let $V_k=(Z_k,X)$ denote the user's observation. Joint privacy gives
\begin{equation}
 (W_{[N]\setminus\{i\}},D_{-k})\ \perp\ V_k\mid D_k=i.
\end{equation}
Under the conditioning $D_k=i$, correctness makes $W_i$ a deterministic function of $V_k$. The pair $(V_k,W_i)$ is therefore also independent of $(W_{[N]\setminus\{i\}},D_{-k})$ under this conditioning. Once $W_i=w_i$ is fixed, the conditional law of $V_k$ cannot depend on the other files or demands. This proves~\eqref{eq:kernel}.

For the cache independence identity, apply the preceding argument to a fixed demand vector $d$ with $d_k=i$. File privacy and correctness imply that $P_{Z_k\mid W=w,D=d}$ depends only on $w_i$. Since placement is independent of the demands, this law equals $P_{Z_k\mid W=w}$ for every $d$. Since $N\ge2$, the same cache law depends only on $w_1$ and, separately, only on $w_2$. The files are independent, so every combination of values in their supports has positive probability. The cache law must therefore be constant in $w$, proving the first identity in~\eqref{eq:independences}.

We next prove that the broadcast is independent of both the library and the demands. First fix a demand vector $d\in[N]^K$ in which two users request distinct files $W_i$ and $W_j$, with $i\ne j$. Applying~\eqref{eq:kernel} to the user requesting $W_i$ shows that $P_{X\mid W=w,D=d}$ depends only on $w_i$. Applying it to the user requesting $W_j$ shows that the same law depends only on $w_j$. Every combination of these file values has positive probability, so the broadcast law is constant in $w$ for this demand vector.

To extend the conclusion to every demand vector, consider $d$ and $d'$ that differ in just one user's request. Since $K\ge2$, another user's request remains fixed. The broadcast marginal of that user's kernel gives
\begin{equation}
 P_{X\mid W=w,D=d}=P_{X\mid W=w,D=d'}.
\end{equation}
Any demand vector can be transformed into any other by changing one user's request at a time. The equality above therefore makes the broadcast law identical for every demand vector. Since the law is independent of $w$ for a vector with two distinct requests, it is independent of $w$ for all vectors. Consequently $I(W_{[N]},D;X)=0$, as claimed.

\section{Proof of the Complementary Anchored Bound}\label{app:dual}
Fix an integer $s\in\{1,\ldots,K\}$, and set $a=N-1$ and $\varepsilon=R-1$. The case $s=1$ is $M\ge1$. For $s\ge2$, we compare two bounds on the joint entropy of broadcasts associated with users $1,\ldots,s$. Let $X_0$ be the broadcast when every user requests $W_1$. For $j=2,\ldots,s$ and $n=2,\ldots,N$, let $X_{j,n}$ be the broadcast when user $j$ requests $W_n$ and every other user requests $W_1$. All of these broadcasts represent alternative demands for the same placement, as in Section~\ref{sec:converse}. For each $j$, let $V_j$ collect the broadcasts associated with user $j$, and let $\cU_{j-1}$ collect $X_0$ and the broadcasts associated with the preceding users:
\begin{equation}
 \begin{aligned}
 V_j&=(X_{j,2},\ldots,X_{j,N}),\\
 \cU_{j-1}&=(X_0,V_2,\ldots,V_{j-1}).
 \end{aligned}
\end{equation}
In particular, $\cU_1=(X_0)$. User $j$ learns no information about $W_{2:N}$ from $(Z_j,X_0)$. It also decodes $W_1$ from that observation, so further conditioning on $W_1$ reveals no information about $W_{2:N}$. The cache $Z_j$ together with $V_j$ decodes all $a$ other files. Therefore
\begin{equation}\label{eq:dual-increment-base}
 \begin{aligned}
 &H(V_j\mid Z_j,W_1,X_0)\\
 &\quad\ge H(W_{2:N}\mid Z_j,W_1,X_0)=aB.
 \end{aligned}
\end{equation}
In every broadcast in $\cU_{j-1}$, user $j$ requests $W_1$. There are $a(j-2)$ such broadcasts besides $X_0$. Conditional subadditivity and~\eqref{eq:delivery-budget} give
\begin{equation}
 H(V_2,\ldots,V_{j-1}\mid Z_j,W_1,X_0)
     \le a(j-2)\varepsilon B.
\end{equation}
Adding the earlier broadcasts $V_2,\ldots,V_{j-1}$ to the conditioning in~\eqref{eq:dual-increment-base} reduces the entropy by at most $a(j-2)\varepsilon B$. Removing $Z_j$ from the conditioning can only increase the entropy. Hence
\begin{equation}
 H(V_j\mid W_1,\cU_{j-1})\ge aB-a(j-2)\varepsilon B.
\end{equation}
Since $X_0$ is independent of $W_1$, summing the increments yields
\begin{align}\label{eq:dual-entropy-lower}
 &H(X_0,V_2,\ldots,V_s\mid W_1)\notag\\
 &\quad\ge H(X_0)+a(s-1)B-a\binom{s-1}{2}\varepsilon B.
\end{align}

We next upper-bound the same joint broadcast entropy using the anchor cache $Z_1$. User 1 requests $W_1$ in every demand under consideration. To cancel the base-broadcast term when the two bounds are compared, we retain $H(X_0)$ explicitly. The identity~\eqref{eq:delivery-exact} gives
\begin{equation}
 \begin{aligned}
 H(X_0\mid Z_1,W_1)&=H(X_0\mid Z_1)-B\\
                  &\le H(X_0)-B.
 \end{aligned}
\end{equation}
By~\eqref{eq:delivery-budget}, each of the remaining $a(s-1)$ broadcasts has entropy at most $\varepsilon B$ conditioned on $(Z_1,W_1)$. Adding $Z_1$ to the broadcast tuple can only increase its entropy conditioned on $W_1$. Conditional subadditivity and $H(Z_1\mid W_1)=H(Z_1)$ then give
\begin{align}\label{eq:dual-entropy-upper}
 &H(X_0,V_2,\ldots,V_s\mid W_1)\notag\\
 &\quad\le H(Z_1)+H(X_0)-B+a(s-1)\varepsilon B.
\end{align}
Comparing~\eqref{eq:dual-entropy-lower} and~\eqref{eq:dual-entropy-upper} cancels $H(X_0)$ and yields
\begin{equation}
 H(Z_1)\ge\left[1+a(s-1)-a\binom s2\varepsilon\right]B.
\end{equation}
Using $H(Z_1)\le MB$ proves Theorem~\ref{thm:dual}.

%% file: sections/09_additional_appendices.tex
\section{Full-Library Secret Sharing With Uniform Privacy Keys}\label{app:spk}
For the comparison in Section~\ref{sec:numerics}, we implement the secret-sharing/private-key (SPK) points of Qi \emph{et al.}~\cite{qi2024} using uniform privacy keys. The implementation supports both repeated and distinct requests. Fix $t\in\{0,\ldots,K-1\}$ and use the dimensions, subset indices, and fixed public matrices from Section~\ref{sec:achievability}.

For each $n\in[N]$, draw $\widetilde U_n$ as a vector of $z$ independent uniform packets in $\F_q^L$. The vectors $\widetilde U_n$ are mutually independent and independent of the library. Encode all $N$ files separately:
\begin{equation}
 \widetilde S_n=PW_n+Q\widetilde U_n,\qquad n\in[N].
\end{equation}
For each row $T\subseteq[K]$ of size $t$, let $\widetilde S_T=(\widetilde S_{n,T}:n\in[N])^{\mathsf T}$. Draw a uniform vector $p_k\in\F_q^N$ for each user $k\in[K]$ and a uniform key packet $T_A\in\F_q^L$ for every $A\subseteq[K]$ of size $t+1$. All private vectors and keys are mutually independent and independent of the library and sharing randomness. User $k$ stores the shares of every file on $\cI_k$, its private vector $p_k$, and one mask for every missing row $T\in\cJ_k$:
\begin{equation}
 V_{k,T}=T_{T\cup\{k\}}+p_k^{\mathsf T}\widetilde S_T.
\end{equation}
For this full-library construction, let $e_1,\ldots,e_N$ denote the standard basis vectors of $\F_q^N$. Publish $q_k=p_k+e_{d_k}$ for each $k\in[K]$, and write $\boldsymbol q=(q_1,\ldots,q_K)$. For every $A\subseteq[K]$ of size $t+1$, transmit
\begin{equation}
 Y_A=T_A+\sum_{j\in A}q_j^{\mathsf T}\widetilde S_{A\setminus\{j\}}.
\end{equation}
User $k\in A$ subtracts the cached interference terms from $Y_A$ and then subtracts $V_{k,A\setminus\{k\}}$. The remaining packet is the missing share $\widetilde S_{d_k,A\setminus\{k\}}$. Combining all missing and cached shares of the requested file recovers that file.

To prove joint privacy, fix an observing user $k$ and its requested index $i$. A simulator draws uniform public controls $\boldsymbol q$ and sets $p_k=q_k-e_i$. Using $W_i$ and freshly drawn sharing randomness, it generates all shares of the requested file. For every other file, the cached shares are a fixed shift of $Q_{\cI_k}\widetilde U_n$ for any fixed library. Since $Q_{\cI_k}$ is invertible and the sharing vectors are independent, the simulator draws those cached shares independently and uniformly. Each cached mask and each broadcast indexed by a set excluding $k$ uses a different independent key. Conditional on all shares, these masks and broadcasts are therefore independent and uniform, so the simulator draws them independently of the generated shares. For each $T\in\cJ_k$, it reconstructs the remaining broadcast from
\begin{equation}
 \begin{aligned}
 Y_{T\cup\{k\}}={}&V_{k,T}+\widetilde S_{i,T}\\
 &+\sum_{j\in T}q_j^{\mathsf T}
                 \widetilde S_{(T\cup\{k\})\setminus\{j\}}.
 \end{aligned}
\end{equation}
Every share in the sum is cached by user $k$, because its row label contains $k$. The simulator uses only $(i,W_i)$ and reproduces the user's full observation for every demand vector with $d_k=i$. This proves joint privacy for repeated and distinct demands alike.

Counting $Nz$ cached share packets, $h$ mask packets, and $\binom K{t+1}$ transmitted packets gives the memory and delivery rate
\begin{equation}\label{eq:s-spk}
 M_t^{\mathrm{SPK}}=1+\frac{Nt}{K-t},\qquad
 R_t^{\mathrm{SPK}}=\frac K{t+1}.
\end{equation}
The private and public controls have vanishing normalized size as the packet length grows. The proposed construction replaces the $N$ full-library share groups by $N-1$ difference-share groups. The aligned correction leaves the delivery rate and mask count unchanged. The memory reduction is therefore exactly $t/(K-t)$ for each $t\in\{0,\ldots,K-1\}$.